\documentclass{sig-alternate}
\pdfoutput=1
\usepackage[ruled,vlined,linesnumbered]{algorithm2e}
\usepackage{amsmath}
\usepackage{amssymb}
\usepackage{tikz}

\newcommand{\ourtitle}{An in-place truncated Fourier transform and applications to polynomial multiplication}
\newcommand{\ourkeywords}{Truncated Fourier transform, fast Fourier transform, polynomial multiplication, in-place algorithms}

\usepackage[letterpaper,a4paper=false,pdftex,%
  colorlinks=false,bookmarks=true,bookmarksnumbered=true,pdfborder={0 0 0},%
  pdftitle={\ourtitle},pdfkeywords={\ourkeywords},%
  pdfauthor={David Harvey, Daniel S. Roche}]{hyperref}

\toappear{Copyright (C) 2010, David Harvey and Daniel S. Roche}

\title{\ourtitle}

\numberofauthors{2}
\author{
\alignauthor
David Harvey\\
  \affaddr{Courant Institute of Mathematical Sciences}\\
  \affaddr{New York University}\\
  \affaddr{New York, New York, U.S.A.}\\
  \email{dmharvey@cims.nyu.edu}
  \email{\href{http://www.cims.nyu.edu/~harvey/}{www.cims.nyu.edu/\~{}harvey/}}
\alignauthor
Daniel S. Roche\\
  \affaddr{Cheriton School of Computer Science}\\
  \affaddr{University of Waterloo}\\
  \affaddr{Waterloo, Ontario, Canada}\\
  \email{droche@cs.uwaterloo.ca}
  \email{\href{http://www.cs.uwaterloo.ca/~droche/}{www.cs.uwaterloo.ca/\~{}droche/}}
}

\newtheorem{thm}{Theorem}[section]
\newtheorem{prop}[thm]{Proposition}
\newtheorem{lem}[thm]{Lemma}
\newdef{remark}{Remark}

\newcommand{\assign}{\leftarrow}
\newcommand{\divides}{\mathbin{|}}
\DeclareMathOperator{\TFT}{InplaceTFT}
\DeclareMathOperator{\FFT}{FFT}
\DeclareMathOperator{\ITFT}{InplaceITFT}
\DeclareMathOperator{\rev}{rev}
\DeclareMathOperator{\len}{len}
\DeclareMathOperator{\fl}{LeftmostLeaf}
\DeclareMathOperator{\rp}{RightmostParent}
\DeclareMathOperator{\il}{IsLeaf}
\DeclareMathOperator{\even}{Even}
\DeclareMathOperator{\odd}{Odd}
\DeclareMathOperator{\parent}{Parent}
\newcommand{\sqb}[1]{[#1]}

\SetKwInOut{Output}{Output}
\SetKw{KwHalt}{halt}
\SetKw{KwOr}{or}
\SetKw{KwTrue}{true}

\begin{document}

\maketitle

\begin{abstract}
The truncated Fourier transform (TFT) was introduced by van der Hoeven
in 2004 as a means of smoothing the ``jumps'' in running time of the ordinary FFT algorithm that occur at power-of-two input sizes. However, the TFT still introduces these jumps in memory usage.
We describe in-place variants of the forward and inverse TFT algorithms, achieving time complexity $O(n \log n)$ with only $O(1)$ auxiliary space. As an application, we extend the second author's results on space-restricted FFT-based polynomial multiplication to polynomials of arbitrary degree.
\end{abstract}

\category{F.2.1}{Analysis of Algorithms and Problem Complexity}%
{Numerical Algorithms and Problems}%
[Computations on polynomials]
\category{G.4}{Mathematical Software}%
{Algorithm design and analysis, Efficiency}
[Algebraic algorithms, Analysis of algorithms]


\section{Introduction}
\subsection{Background}
The discrete Fourier transform (DFT) is a linear map that evaluates a given polynomial at powers of a root of unity. Cooley and Tukey \cite{CooTuk65} were the first to develop an
efficient method to compute this transform on a digital computer, 
known as the Fast Fourier
Transform (FFT). This algorithm has since become one of the most important and
useful tools in computer science, especially in the area of signal
processing.

The FFT algorithm is also important in computer algebra, most
notably in asymptotically fast methods for integer and polynomial
multiplication. The first integer multiplication algorithm to run in
softly linear time
relies on the FFT \cite{SchStr71}, as do the recent theoretical
improvement \cite{Fur07} and the best result for polynomial
multiplication over arbitrary algebras \cite{CanKal91}.
Moreover, numerous other operations on polynomials --- including division,
evaluation/interpolation, and GCD computation --- have been reduced to
multiplication, so more efficient multiplication methods have an
indirect effect on many areas in computer algebra 
\cite[\S 8--11]{vzgGer03}.

The simplest FFT to implement, and the fastest in practice, is
the radix-2 Cooley-Tukey FFT. Because the radix-2 FFT requires the size to be a power of two, 
the simplest solution for all other sizes is to pad the input polynomials
with zeros, resulting in large unwanted
``jumps'' in the complexity at powers of two.

\subsection{The truncated Fourier transform}

It has been known for some time that if only a subset of the DFT output is needed, then the FFT can be truncated or ``pruned'' to reduce the complexity, essentially by disregarding those parts of the computation tree not contributing to the desired outputs \cite{Mar71,SorBur93}. More recently, van der Hoeven took the crucial step of showing how to invert this process, describing a truncated Fourier transform (TFT) and an \emph{inverse} truncated Fourier transform (ITFT), and showing that this leads to a polynomial multiplication algorithm whose running time varies relatively smoothly in the input size \cite{vdh-1,vdh-2}. 

Specifically, given an input vector of length $n \leq 2^k$, the TFT computes the first $n$ coefficients of the ordinary Fourier transform of length $2^k$, and the ITFT computes the inverse of this map. The running time of these algorithms smoothly interpolates the $O(n \log n)$ complexity of the standard radix-$2$ Cooley--Tukey FFT algorithm. One can therefore deduce an asymptotically fast polynomial multiplication algorithm that avoids the characteristic ``jumps'' in running time exhibited by traditional FFT-based polynomial multiplication algorithms when the output degree crosses a power-of-two boundary. This observation has been confirmed with practical implementations \cite{vdh-2,LiMazSch09,cache-trunc-fft}, with the most marked improvements in the multivariate case.

One drawback of van der Hoeven's algorithms is that while their time complexity varies smoothly with $n$, their space complexity does not. Both the TFT and ITFT operate in a buffer of length $2^{\lceil \lg n \rceil}$; that is, for inputs of length $n$, they require auxiliary storage of $2^{\lceil \lg n \rceil} - n + O(1)$ cells to store intermediate results, which can be $\Omega(n)$ in the worst case. 

\subsection{Summary of results}

The main results of this paper are TFT and ITFT algorithms that require only $O(1)$ auxiliary space, while respecting the $O(n \log n)$ time bound.

The new algorithms have their origin in a cache-friendly variant of the TFT and ITFT given by the first author \cite{cache-trunc-fft}, that builds on Bailey's cache-friendly adaptation of the ordinary FFT \cite{bailey}. If the transform takes place in a buffer of length $L = 2^\ell$, these algorithms decompose the transform into $L_1 = 2^{\ell_1}$ row transforms of length $L_2 = 2^{\ell_2}$ and $L_2$ column transforms of length $L_1$, where $\ell_1 + \ell_2 = \ell$. Van der Hoeven's algorithms correspond to the case $L_1 = 2$ and $L_2 = L/2$. To achieve optimal locality, \cite{cache-trunc-fft} suggests taking $L_i \approx \sqrt L$ ($\ell_i \approx \ell/2$). In fact, in this case one already obtains TFT and ITFT algorithms needing only $O(\sqrt n)$ auxiliary space. At the other extreme we may take $L_1 = L/2$ and $L_2 = 2$, obtaining TFT and ITFT algorithms that use only $O(1)$ space at each recursion level, or $O(\log n)$ auxiliary space altogether. In signal processing language, these may be regarded as decimation-in-time variants of van der Hoeven's decimation-in-frequency algorithms.

Due to data dependencies in the $O(\log n)$-space algorithms sketched above, the space usage cannot be reduced further by simply reordering the arithmetic operations. In this paper, we show that with a little extra work, increasing the implied constant in the $O(n \log n)$ running time bound, it is possible to reduce the auxiliary space to only $O(1)$. To make the $O(1)$ space bound totally explicit, we present our TFT and ITFT algorithms (Algorithms \ref{algo:tft} and \ref{algo:itft}) in an iterative fashion, with no recursion. Since we do not have space to store all the necessary roots of unity, we explicitly include steps to compute them on the fly; this is non-trivial because the decimation-in-time approach requires indexing the roots in bit-reversed order.

As an application, we generalize the second author's space-restricted polynomial multiplication algorithm \cite{roche}. Consider a model in which the input polynomials are considered read-only, but the output buffer may be read from and written to multiple times. The second author showed that in such a model, it is possible to multiply polynomials of degree $n = 2^k - 1$ in time $O(n \log n)$ using only $O(1)$ auxiliary space. Using the new in-place ITFT, we generalize this result to polynomials of arbitrary degree.

\section{Preliminaries}

\subsection{Computational model}

We work over a ring $R$ containing $2^k$-th roots of unity for all (suitably large) $k$, and in which $2$ is not a zero-divisor. 

Our memory model is similar to that used in the study of in-place algorithms for sorting and geometric problems, combined with the well-studied notion of algebraic complexity. Specifically, we allow two primitive types in memory: ring elements and pointers. A ring element is any single element of $R$, and the input to any algorithm will consist of $n$ such elements stored in an array. A pointer can hold a single integer $a\in\mathbb{Z}$ in the range $-cn \leq a \leq cn$ for some fixed constant $c\in\mathbb{N}$. (In our algorithms, we could take $c=2$.)

We say an algorithm is in-place if it overwrites its input buffer with
the output. In this case, any element in this single input/output array
may be read from or written to in constant time. Our in-place truncated
Fourier transform algorithms (Algorithms \ref{algo:tft} and
\ref{algo:itft}) fall under this model.

An out-of-place algorithm uses separate memory locations for input and
output. Here, any element from the input array may be read from in
constant time (but not overwritten), and any element in the output array
may be read from or written to in constant time as well. This will be
the situation in our multiplication algorithm (Algorithm
\ref{algo:multiply}).

The algorithms also need to store some number of pointers and ring
elements not in the input or output arrays, which we define to be the
auxiliary storage used by the algorithm. All the algorithms we present
will use only $O(1)$ auxiliary storage space.

This model should correspond well with practice, at least when the computations are performed in main memory and the ring $R$ is finite.

\subsection{DFT notation}

We denote by $\omega_{[k]}$ a primitive $2^k$-th root unity, and we assume that these are chosen compatibly, so that $\omega_{[k+1]}^2 = \omega_{[k]}$ for all $k \geq 0$. Define a sequence of roots $\omega_0, \omega_1, \ldots$ by $\omega_s = \omega_{[k]}^{\rev_k s}$, where $k \geq \lceil \lg (s+1) \rceil$ and $\rev_k s$ denotes the length-$k$ bit-reversal of $s$. Thus we have
\begin{align*}
 \omega_0 & = \omega_{[0]} \text{ ($= 1$)}  &
 \omega_2 & = \omega_{[2]} &
 \omega_4 & = \omega_{[3]} &
 \omega_6 & = \omega_{[3]}^3 \\
 \omega_1 & = \omega_{[1]} \text{ ($= -1$)} &
 \omega_3 & = \omega_{[2]}^3 &
 \omega_5 & = \omega_{[3]}^5 &
 \omega_7 & = \omega_{[3]}^7
\end{align*}
and so on. Note that
 \[ \omega_{2s+1} = -\omega_{2s} \qquad \text{and} \qquad \omega_{2s}^2 = \omega_{2s+1}^2 = \omega_s. \]

If $F \in R[x]$ is a polynomial with $\deg F < n$, we write $F_s$ for the coefficient of $x^s$ in $F$, and we define the Fourier transform $\hat F$ by
 \[ \hat F_s = F(\omega_s). \]
In Algorithms \ref{algo:tft} and \ref{algo:itft} below, we decompose $F$ as
 \[ F(x) = G(x^2) + x H(x^2), \]
where $\deg G < \lceil n/2 \rceil$ and $\deg H < \lfloor n/2 \rfloor$. Using the properties of $\omega_s$ mentioned above, we obtain the ``butterfly'' relations
\begin{equation}
\label{eq:butterfly}
\begin{aligned}
 \hat F_{2s} & = \hat G_s + \omega_{2s} \hat H_s, \\
 \hat F_{2s+1} & = \hat G_s - \omega_{2s} \hat H_s.
\end{aligned}
\end{equation}

Both the TFT and ITFT algorithm require, at each recursive level, iterating
through a set of index-root pairs such as $\{(i,\omega_i), 0 \leq i < n\}$. A
traditional, time-efficient approach would be to precompute all powers of
$\omega_{[k]}$, store them in reverted-binary order, and then pass
through this array with a single pointer. However, this is impossible under the restriction that no auxiliary
storage space be used. Instead, we will compute the roots on-the-fly by
iterating through the powers of $\omega_{[k]}$ in order, and through the
indices $i$ in bit-reversed order. Observe that incrementing an integer
counter through $\rev_k 0, \rev_k 1, \rev_k 2, \ldots$ can be done in exactly
the same way as incrementing through $0,1,2,\ldots$, which is possible
in-place and in amortized constant time.

\section{Space-restricted TFT}
\label{sec:tft}

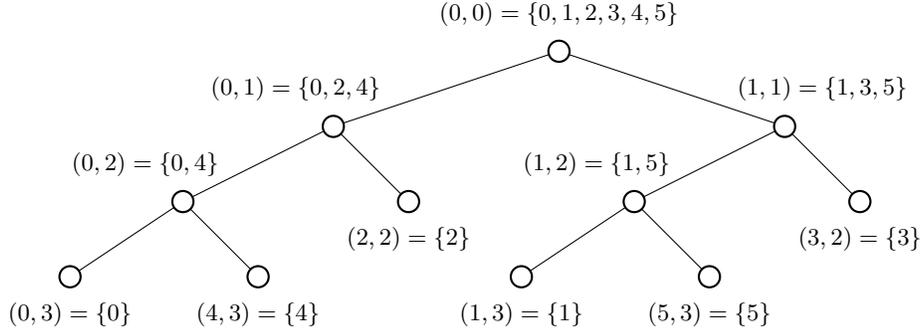
\begin{figure*}
\begin{center}
\begin{tikzpicture}[scale=1.0]
\draw (0,0) -- (-3,-1);
\draw (0,0) -- (3,-1);
\draw (-3,-1) -- (-5,-2);
\draw (-3,-1) -- (-2,-2);
\draw (3,-1) -- (1,-2);
\draw (3,-1) -- (4,-2);
\draw (-5,-2) -- (-6.5,-3);
\draw (-5,-2) -- (-4,-3);
\draw (1,-2) -- (-0.5,-3);
\draw (1,-2) -- (2,-3);
\draw[thick,fill=white] (0,0) circle (4pt);
\draw (0,0.5) node {$(0,0) = \{0,1,2,3,4,5\}$};
\draw[thick,fill=white] (-3,-1) circle (4pt);
\draw (-3.5,-0.5) node {$(0,1) = \{0,2,4\}$};
\draw[thick,fill=white] (3,-1) circle (4pt);
\draw (3.5,-0.5) node {$(1,1) = \{1,3,5\}$};
\draw[thick,fill=white] (-5,-2) circle (4pt);
\draw (-5.5,-1.5) node {$(0,2) = \{0,4\}$};
\draw[thick,fill=white] (-2,-2) circle (4pt);
\draw (-2,-2.5) node {$(2,2) = \{2\}$};
\draw[thick,fill=white] (1,-2) circle (4pt);
\draw (0.5,-1.5) node {$(1,2) = \{1,5\}$};
\draw[thick,fill=white] (4,-2) circle (4pt);
\draw (4,-2.5) node {$(3,2) = \{3\}$};
\draw[thick,fill=white] (-6.5,-3) circle (4pt);
\draw (-6.5,-3.5) node {$(0,3) = \{0\}$};
\draw[thick,fill=white] (-4,-3) circle (4pt);
\draw (-4,-3.5) node {$(4,3) = \{4\}$};
\draw[thick,fill=white] (-0.5,-3) circle (4pt);
\draw (-0.5,-3.5) node {$(1,3) = \{1\}$};
\draw[thick,fill=white] (2,-3) circle (4pt);
\draw (2,-3.5) node {$(5,3) = \{5\}$};
\end{tikzpicture}
\end{center}
\caption{TFT tree for $n = 6$}
\label{fig:tree}
\end{figure*}

In this section we describe an in-place TFT algorithm that uses only $O(1)$
auxiliary space (Algorithm \ref{algo:tft}). The routine operates on a
buffer $X_0, \ldots, X_{n-1}$ containing elements of $R$. It takes as
input a root of unity of sufficiently high order and the coefficients $F_0, \ldots, F_{n-1}$ of a polynomial $F \in R[x]$, and overwrites these with $\hat F_0, \ldots, \hat F_{n-1}$.

The pattern of the algorithm is recursive, but we avoid recursion by explicitly moving through the recursion tree, avoiding unnecessary space usage. An example tree for $n = 6$ is shown in Figure \ref{fig:tree}. The node $S = (q, r)$ represents the subarray with offset $q$ and stride $2^r$; the $i$th element in this subarray is
$S_i = X_{q+i\cdot 2^r}$, and the length of the subarray is given by
$$\len(S) = \left\lceil \frac{n-q}{2^r} \right\rceil.$$

The root is $(0, 0)$, corresponding to the entire input array of length $n$. Each subarray of length $1$ corresponds to a leaf node, and we define the predicate $\il(S)$ to be true iff $\len(S)=1$. Each non-leaf
node splits into even and odd child nodes. To facilitate the path
through the tree, we define
\begin{align*}
\even(q,r) &= (q,r+1), \\
\odd(q,r) &= (q+2^r, r+1)
\end{align*}
if $(q, r)$ is not a leaf,
$$ \parent(q,r) = \begin{cases}
   (q,r-1), & q < 2^{r-1}, \\ 
   (q-2^{r-1},r-1), & q \geq 2^{r-1} \end{cases} $$
if $(q, r)$ is not the root, and for any node we define
$$ \fl(S) = \begin{cases}
   S, & \il(S), \\
   \fl(\even(S)), & \textrm{otherwise.} \end{cases}
$$

\begin{algorithm}[ht]
\DontPrintSemicolon
\KwIn{$X_i = F_i$ for $0 \leq i < n$, where $F \in R[x]$, $\deg F < n$}
\KwOut{$X_i = \hat F_i$ for $0 \leq i < n$}
\medskip
$S \assign \fl(0,0)$\;
$\textit{prev} \assign \textrm{null}$\;
\While{\KwTrue}{\label{line:tft-while}
  $m \assign \len(S)$\;
  \If{$\il(S)$ \KwOr $\textit{prev} = \odd(S)$}{
    \For{$(i,\theta) \in \{(j,\omega_{2j}):0\leq j < \lfloor m/2\rfloor\}$
    \label{line:tftbfly1}}{
      $\left[\begin{array}{c}S_{2i}\\S_{2i+1}\end{array}\right] \assign
       \left[\begin{array}{c}S_{2i} + \theta S_{2i+1} \\
       S_{2i} - \theta S_{2i+1}\end{array}\right]$\;
       \label{line:tftbfly2}
    }
    \lIf{$S = (0,0)$}{\KwHalt}\label{line:tfthalt}\;
    $\textit{prev} \assign S$\;
    $S \assign \parent(S)$\;
  }
  \ElseIf{$\textit{prev} = \even(S)$}{
    \If{$\len(S) \equiv 1 \bmod 2$ \label{line:tftodd}}{
      $v \gets \sum_{i=0}^{(m-3)/2} 
       S_{2i+1}\cdot (\omega_{(m-1)/2})^i$\label{line:tftv}\;
      $S_{m-1} \gets S_{m-1} + v \cdot \omega_{m-1}$
      \label{line:tftcorrect}\;
    }
    $\textit{prev} \assign S$\;
    $S \assign \fl(\odd(S))$ \label{line:tftgodown}\;
  }
}
\caption{$\TFT(\sqb{X_0, \ldots, X_{n-1}})$}
\label{algo:tft}
\end{algorithm}

We begin with the following lemma.
\begin{lem}\label{lem:tft}
Let $N$ be a node with $\len(N)=\ell$, 
and let \[ A(x) = \sum_{0\leq i<\ell} A_i x^i \in R[x]. \]
If $S = \fl(N)$ and $N_i = A_i$ for $0\leq i<\ell$
before some iteration of line~\ref{line:tft-while}
in Algorithm~\ref{algo:tft}, then after a finite number of steps, 
we will have $S=N$ and $N_i = \hat{A}_i$ for $0 \leq i < \ell$,
before the execution of line~\ref{line:tfthalt}.
No other array entries in $X$ are affected.
\end{lem}
\begin{proof}
The proof is by induction on $\ell$. If $\ell=1$, then $\il(N)$ is
true and $\hat{A}_0 = A_0$ so we are done. So assume 
$\ell>1$ and that the lemma holds for all shorter lengths. 

Decompose $A$ as $A(x) = G(x^2) + xH(x^2)$. Since 
$S = \fl(\even(N))$ as well, the induction hypothesis guarantees that
the even-indexed elements of $N$, corresponding to the coefficients of
$G$, will be transformed into $\hat{G}$, 
and we will have $S=\even(N)$ before line~\ref{line:tfthalt}. The following
lines set $prev=\even(N)$ and $S=N$, so that
lines~\ref{line:tftodd}--\ref{line:tftgodown} are executed on the next
iteration. 

If $\ell$ is odd, then
$(\ell-1)/2 \geq \len(\odd(N))$, so $\hat{H}_{(\ell-1)/2}$ will not be
computed in the odd subtree, and we will not be able to apply
\eqref{eq:butterfly} to compute
$\hat{A}_{\ell-1} = \hat{G}_{(\ell-1)/2} + \omega_{\ell-1}\hat{H}_{(\ell-1)/2}$.
This is why, in this case, we explicitly 
compute $$v = H(\omega_{(\ell-1)/2}) = \hat{H}_{(\ell-1)/2}$$ 
on line~\ref{line:tftv}, and then
compute $\hat{A}_{\ell-1}$ directly on line~\ref{line:tftcorrect}, before
descending into the odd subtree.

Another application of the induction hypothesis guarantees that we will
return to line~\ref{line:tfthalt} with $S=\odd(N)$ after computing $N_{2i+1} = \hat H_i$ for $0 \leq i < \lfloor \ell/2 \rfloor$.
The following lines set $prev=\odd(N)$ and $S=N$, and we arrive at
line~\ref{line:tftbfly1} on the next iteration.
The \textbf{for} loop
thus properly applies the butterfly relations \eqref{eq:butterfly} to
compute $\hat{A}_i$ for $0 \leq i < 2\lfloor \ell/2 \rfloor$,
which completes the proof.
\end{proof}

Now we are ready for the main result of this section.
\begin{prop}
\label{prop:tft}
Algorithm \ref{algo:tft} correctly computes $\hat F_i$ for $0 \leq i < n$.
It performs $O(n \log n)$ ring and pointer operations, 
and uses $O(1)$ auxiliary space.
\end{prop}
\begin{proof}
The correctness follows immediately from Lemma \ref{lem:tft}, 
since we start with $S=\fl(0,0)$, which is the first leaf of the whole
tree. The space bound is immediate since each variable
has constant size.

To verify the time bound, notice that the {\bf while} loop visits each leaf node once and each non-leaf node twice (once with $\textit{prev} = \even(S)$ and once with $\textit{prev} = \odd(S)$). Since always $q < 2^r < 2n$, there are $O(n)$ iterations through the {\bf while} loop, each of which has cost 
$O(\len(S) + \log n)$. This gives the total cost of $O(n\log n)$.
\end{proof}

\section{Space-restricted ITFT}
\label{sec:itft}

Next we describe an in-place inverse TFT algorithm that uses $O(1)$
auxiliary space (Algorithm \ref{algo:itft}). It takes as input $\hat F_0, \ldots, \hat F_{n-1}$ for some polynomial $F \in R[x]$, $\deg F < n$, and overwrites the buffer with $F_0, \ldots, F_{n-1}$.

The path of the algorithm is exactly the reverse of
Algorithm~\ref{algo:tft}, and we use the same notation as before to move
through the tree. We only require one additional function:
\setlength{\multlinegap}{0pt}
\begin{multline*}
\rp(S) = \\
\begin{cases}
S, & S = \odd(\parent(S)), \\
\rp(\parent(S)), & \text{otherwise}.
\end{cases}
\end{multline*}
If \[ \fl(\odd(N_1)) = N_2, \] then
\[ \parent(\rp(N_2)) = N_1, \] so $\rp$ computes the inverse of the assignment on line~\ref{line:tftgodown} in
Algorithm~\ref{algo:tft}.

\begin{algorithm}[ht]
\DontPrintSemicolon
\KwIn{$X_i = \hat F_i$ for $0 \leq i < n$, where $F \in R[x]$, $\deg F < n$}
\KwOut{$X_i = F_i$ for $0 \leq i < n$}
\medskip
$S \assign (0,0)$\;
\While{$S \neq \fl(0,0)$\label{line:itft-while}}{
  \If{$\il(S)$}{
    $S \assign \parent(\rp(S))$\;
    $m \assign \len(S)$\;
    \If{$\len(S) \equiv 1 \bmod 2$ \label{line:itftodd}}{
      $v \gets \sum_{i=0}^{(m-3)/2} 
       S_{2i+1}\cdot \omega_{(m-1)/2}^i$\label{line:itftv}\;
      $S_{m-1} \gets S_{m-1} - v \cdot \omega_{m-1}$
      \label{line:itftcorrect}
    }
    $S \assign \even(S)$\;
  }
  \Else{
    $m \assign \len(S)$\;
    \For{$(i,\theta) \in \{(j,\omega_{2j}^{-1}):0\leq j < \lfloor m/2\rfloor\}$
    \label{line:itftbfly1}}{
      $\left[\begin{array}{c}S_{2i}\\S_{2i+1}\end{array}\right] \assign
       \left[\begin{array}{c}(S_{2i} + S_{2i+1})/2 \\
       \theta\cdot(S_{2i} - S_{2i+1})/2\end{array}\right]$\;
       \label{line:itftbfly2}
    }
    $S \assign \odd(S)$\label{line:itftlast}\;
  }
}
\caption{$\ITFT(\sqb{X_0, \ldots, X_{n-1}})$}
\label{algo:itft}
\end{algorithm}

We leave it to the reader to confirm that the structure of the recursion is
identical to that of Algorithm~\ref{algo:tft}, but in reverse, from which the following analogues of Lemma~\ref{lem:tft} and Proposition~\ref{prop:tft} follow immediately:

\begin{lem}\label{lem:itft}
Let $N$ be a node with $\len(N)=\ell$, 
and \[ A(x) = \sum_{0\leq i<\ell} A_i x^i \in R[x]. \]
If $S=N$ and $N_i=\hat{A}_i$ for $0 \leq i < \ell$
before some iteration of line~\ref{line:itft-while} in
Algorithm~\ref{algo:itft}, then after a finite number of steps, 
we will have $S=\fl(N)$ and $N_i = A_i$ for $0 \leq i < \ell$
before some iteration of line~\ref{line:itft-while}.
No other array entries in $X$ are affected.
\end{lem}

\begin{prop}
\label{prop:itft}
Algorithm \ref{algo:itft} correctly computes $F_i$ for $0 \leq i < n$. It
performs $O(n \log n)$ ring and pointer operations, and uses $O(1)$ auxiliary space.
\end{prop}

The fact that our $\TFT$ and
$\ITFT$ algorithms are essentially reverses of each other is an
interesting feature not shared by the original formulations in
\cite{vdh-1}.

\section{Polynomial multiplication}

We now describe the multiplication algorithm alluded to in the introduction. The strategy is similar to that of \cite{roche}, with a slightly more complicated ``folding'' step. The input consists of two polynomials $A, B \in R[x]$ with $\deg A < n$ and $\deg B < m$. The routine is supplied an output buffer $X$ of length $r = n + m - 1$ in which to write the product $C = AB$.

The subroutine $\FFT$ has the same interface as $\TFT$, but is only called for power-of-two length inputs.

\begin{algorithm}
\DontPrintSemicolon
\KwIn{$A, B \in R[x]$, $\deg A < m$, $\deg B < n$}
\KwOut{$X_s = C_s$ for $0 \leq s < n + m - 1$, where $C = AB$}
\medskip
$r \assign n + m - 1$\;
$q \assign 0$\;
\While{$q < r - 1$}{
  $\ell \assign \lfloor \lg(r-q) \rfloor - 1$\;
  $L \assign 2^\ell$\;
  $[X_{q},X_{q+1},\ldots,X_{q+2L-1}] \assign [0,0,\ldots,0]$
  \nllabel{line:fold-0}\;
  \For{$0 \leq i < m$\nllabel{line:fold-1}}{$X_{q + (i \bmod L)} \assign X_{q + (i \bmod L)} + \omega_q^i A_i$\nllabel{line:fold-1b}}
  $\FFT([X_q, X_{q+1}, \ldots, X_{q + L - 1}])$\nllabel{line:fft-1}\;
  \For{$0 \leq i < n$\nllabel{line:fold-2}}{$X_{q + L + (i \bmod L)} \assign X_{q + L + (i \bmod L)} + \omega_q^i B_i$\nllabel{line:fold-2b}}
  $\FFT([X_{q+L}, X_{q+L+1}, \ldots, X_{q + 2L - 1}])$\;
  \For{$0 \leq i < L$\nllabel{line:mult}}{$X_{q+i} \assign X_{q+i} X_{q+L+i}$\nllabel{line:multb}}
  $q \assign q + L$\;
}
$X_{r-1} \assign A(\omega_{r-1}) B(\omega_{r-1})$\nllabel{line:mult-last}\;
$\ITFT([X_0, \ldots, X_{r-1}])$\nllabel{line:mult-itft}\;
\caption{Space-restricted product}
\label{algo:multiply}
\end{algorithm}

\begin{prop}
Algorithm \ref{algo:multiply} correctly computes the product $C = AB$, in time
$O((m+n) \log (m+n))$ and using $O(1)$ auxiliary space.
\end{prop}
\begin{proof}
The main loop terminates since $q$ is strictly increasing. Let $N$ be the number of iterations, and let $q_0 > q_1 > \cdots > q_{N-1}$ and $L_0 \geq L_1 \geq \cdots \geq L_{N-1}$ be the values of $q$ and $L$ on each iteration. By construction, the intervals $[q_i, q_i + L_i)$ form a partition of $[0, r - 1)$, and $L_i$ is the largest power of two such that $q_i + 2L_i \leq r$. Therefore each $L$ can appear at most twice (i.e.~if $L_i = L_{i-1}$ then $L_{i+1} < L_i$), $N \leq 2 \lg r$, and we have $L_i \divides q_i$ for each $i$.

At each iteration, lines \ref{line:fold-1}--\ref{line:fold-1b} compute the coefficients of the polynomial $A(\omega_q x) \bmod x^L - 1$, placing the result in $[X_q, \ldots, X_{q+L-1}]$. Line \ref{line:fft-1} then computes $X_{q+i} = A(\omega_q \omega_i)$ for $0 \leq i < L$. Since $L \divides q$ we have $\omega_q \omega_i = \omega_{q+i}$, and so we have actually computed $X_{q+i} = \hat A_{q+i}$ for $0 \leq i < L$. The next two lines similarly compute $X_{q+L+i} = \hat B_{q+i}$ for $0 \leq i < L$. (The point of the condition $q + 2L \leq r$ is to ensure that both of these transforms fit into the output buffer.) Lines \ref{line:mult}--\ref{line:multb} then compute $X_{q+i} = \hat A_{q+i} \hat B_{q+i} = \hat C_{q+i}$ for $0 \leq i < L$.

After line \ref{line:mult-last} we finally have $X_s = \hat C_i$ for all $0 \leq i < r$. (The last product was handled separately since the output buffer does not have room for the two Fourier coefficients.) Line \ref{line:mult-itft} then recovers $C_0, \ldots, C_{r-1}$.

We now analyze the time and space complexity. The loops on lines \ref{line:fold-0}, \ref{line:fold-1}, \ref{line:fold-2} and \ref{line:mult} contribute $O(r)$ operations per iteration, or $O(r \log r)$ in total, since $N = O(\log r)$. The FFT calls contribute $O(L_i \log L_i)$ per iteration, for a total of $O(\sum_i L_i \log L_i) = O(\sum_i L_i \log L) = O(r \log r)$. Line \ref{line:mult-last} contribute $O(r)$, and line \ref{line:mult-itft} contributes $O(r \log r)$ by Proposition \ref{prop:itft}. The space requirements are immediate also by Proposition \ref{prop:itft}, since the main loop requires only $O(1)$ space.
\end{proof}

\section{Conclusion}

We have demonstrated that forward and inverse
radix-2 truncated Fourier transforms can be
computed in-place using $O(n \log n)$ time and $O(1)$ auxiliary storage.
As a result, polynomials with degrees less than $n$ can be multiplied
out-of-place within the same time and space bounds. These results apply
to any size $n$, whenever 
the underlying ring 
admits division by 2 and a primitive root of unity of order
$2^{\lceil \lg n \rceil}$.

Numerous questions remain open in this direction. First, our in-place
TFT and ITFT algorithms avoid using auxiliary space at the cost of some extra arithmetic. So although
the asymptotic complexity is still $O(n\log n)$, the implied constant
will be greater than for the usual TFT or FFT algorithms. It would be
interesting to know whether this extra cost is unavoidable. In
any case, the implied constant would need to be reduced as much as
possible for the in-place TFT/ITFT to compete with the running
time of the original algorithms.

We also have not yet demonstrated an in-place multi-dimensional TFT or ITFT
algorithm. In one dimension, the ordinary TFT can hope to gain at most a factor of two over the FFT, but a $d$-dimensional TFT can be faster than the corresponding FFT by a factor of $2^d$, as demonstrated in \cite{LiMazSch09}. An in-place variant along the lines of the algorithms presented in this paper could save a factor of $2^d$ in both time and memory, with practical consequences for multivariate polynomial arithmetic.

Finally, noticing that our multiplication algorithm, despite using only
$O(1)$ auxiliary storage, is still an out-of-place algorithm, we
restate an open question of \cite{roche}: Is it possible,
under any time restrictions, to perform multiplication in-place and
using only $O(1)$ auxiliary storage? The answer seems to be no, but a
proof is as yet elusive.

\bibliographystyle{hplain}
\bibliography{inplace-tft}

\end{document}